\documentclass[journal]{IEEEtran}
\IEEEoverridecommandlockouts

\usepackage{amsmath}
\usepackage{amsfonts}
\usepackage{amssymb}
\usepackage{amsthm}
\usepackage{tabularx}
\usepackage{mathrsfs} 

\usepackage{url}
\usepackage{hyperref}
\newtheorem{theorem}{Theorem}

\newtheorem{lemma}{Lemma}

\usepackage{stfloats}
\usepackage{float}
\usepackage{graphicx}
\usepackage{cite}
\usepackage{xcolor}
\usepackage{subfigure}
\renewcommand{\vec}[1]{\mathbf{#1}}

\makeatletter
\def\blfootnote{\xdef\@thefnmark{}\@footnotetext}
\makeatother
\def\BibTeX{{\rm B\kern-.05em{\sc i\kern-.025em b}\kern-.08em
    T\kern-.1667em\lower.7ex\hbox{E}\kern-.125emX}}

\begin{document}

\title{Physical Layer Security in FAS-aided Wireless Powered NOMA Systems
}
\author{Farshad~Rostami~Ghadi,~\IEEEmembership{Member, IEEE},~Masoud~Kaveh,~Kai-Kit~Wong,~\IEEEmembership{Fellow},~\textit{IEEE}, Diego~Mart\'in, Riku~Jäntti,~\IEEEmembership{Senior Member},~\textit{IEEE}, and~Zheng~Yan,~\IEEEmembership{Fellow},~\textit{IEEE}}


	\maketitle
 
	\begin{abstract}
    The rapid evolution of communication technologies and the emergence of sixth-generation (6G) networks have introduced unprecedented opportunities for ultra-reliable, low-latency, and energy-efficient communication. However, the integration of advanced technologies like non-orthogonal multiple access (NOMA) and wireless powered communication networks (WPCNs) brings significant challenges, particularly in terms of energy constraints and security vulnerabilities. Traditional antenna systems and orthogonal multiple access schemes struggle to meet the increasing demands for performance and security in such environments.
To address this gap, this paper investigates the impact of emerging fluid antenna systems (FAS) on the performance of physical layer security (PLS) in WPCNs. Specifically, we consider a scenario in which a transmitter, powered by a power beacon via an energy link, transmits confidential messages to legitimate FAS-aided users over information links while an external eavesdropper attempts to decode the transmitted signals. Additionally, users leverage the NOMA scheme, where the far user may also act as an internal eavesdropper. For the proposed model, we first derive the distributions of the equivalent channels at each node and subsequently obtain compact expressions for the secrecy outage probability (SOP) and average secrecy capacity (ASC), using the Gaussian quadrature methods. Our results reveal that incorporating the FAS for NOMA users, instead of the TAS, enhances the performance of the proposed secure WPCN.
	\end{abstract}
 
	\maketitle
	
    \blfootnote{The work of F. R. Ghadi and K. K. Wong is supported by the Engineering and Physical Sciences Research Council (EPSRC) under Grant EP/W026813/1. 
    The work of M. Kaveh, R. Jäntti, and Z. Yan is supported in part by the Academy of Finland under Grants 345072 and 350464.
    The work of D. Mart\'in has been developed within the project PRESECREL (PID2021-124502OB-C43). We would like to acknowledge the financial support of the Ministerio de Ciencia e Investigación (Spain), in relation to the Plan Estatal de Investigación Científica y Técnica y de Innovación 2017-2020.}
\blfootnote{\noindent F. Rostami Ghadi and K. K. Wong are with the Department of Electronic and Electrical Engineering, University College London, London, UK. K. K. Wong is also affiliated with the Department of Electronic Engineering, Kyung Hee University, Yongin-si, Gyeonggi-do 17104, Korea. (e-mail:$\{\rm f.rostamighadi,kai\text{-}kit.wong\}@ucl.ac.uk$).}
\blfootnote{\noindent M. Kaveh and R. Jäntti are with with the Department of Information and Communication Engineering, Aalto University, 02150 Espoo, Finland. (e-mail: $\rm \{masoud.kaveh, riku.jantti\}@aalto.fi$).}
\blfootnote{\noindent D. Mart\'in is with the Department of Computer Science, Escuela de Ingeniería Informática de Segovia, Universidad de Valladolid, Segovia, 40005, Spain (e-mail:$\rm diego.martin.andres@uva.es$).}
\blfootnote{Z. Yan is with the School of Cyber Engineering Xidian University, Xi'an, China. (e-mail:$\rm zyan@xidian.edu.cn$).}

	\vspace{0mm}
	\section{Introduction}\label{sec-intro}
	\IEEEPARstart{T}{he}
rapid advancements in communication technologies and the emergence of the sixth generation (6G) network promise unprecedented capabilities, including ultra-high data rates, massive connectivity, and near-zero latency. However, these advances come with significant energy challenges, particularly for devices that are energy-constrained \cite{zhang2022}. In 6G networks, the densification of base stations, the proliferation of edge devices, and the continuous demand for high-speed communication place tremendous strain on energy resources. Traditional battery powered devices face inherent limitations, such as frequent recharging, short operating lifetimes, and the environmental impact of large-scale battery production and disposal \cite{moloudian2024}. These constraints threaten the feasibility of deploying energy-efficient and sustainable networks on the scale envisioned for 6G.

To overcome these challenges, innovative solutions are required to ensure reliable and perpetual operation of devices without reliance on conventional battery technologies. wireless powered communication networks (WPCNs) have emerged as a promising solution by integrating wireless power transfer (WPT) with communication systems \cite{bi2016,wang2023,huang2022}. WPCNs enable energy-constrained devices to harvest energy wirelessly while maintaining data communication, thus eliminating the dependence on batteries for continuous operation \cite{mahmood2022}. This approach is particularly well suited to support the demands of 6G networks, where scalability, sustainability, and uninterrupted connectivity are paramount.
By combining energy transfer with data communication, WPCNs offer a sustainable and efficient framework to support a wide range of energy-constrained applications, such as the Internet of Things (IoT) \cite{chen2019,vu2021,do2020}, backscatter communication (BC) \cite{han2017,ghadi2022,lyu2017,kaveh2023}, sensor networks \cite{ijemaru2022}, smart grids \cite{liu2023,kaveh_ris2023}, vehicle-to-vehicle (V2V) \cite{shafiqurrahman2023}, and remote monitoring systems \cite{li2022}. 

On the other hand, multi-input multi-output (MIMO) systems face challenges such as high complexity in hardware and signal processing, limited spatial diversity in compact devices, and performance degradation in dynamic environments with poor channel conditions \cite{rajak2022}. 
Fluid antenna systems (FASs)  have shown promising potential to address these by dynamically adapting antenna positions to optimize signal reception, enabling higher diversity gains and better performance in fluctuating environments, while reducing the complexity associated with traditional multi-antenna setups \cite{wongfluid2021,wong2perf2020,wang2024}. 
FAS leverages fluidic or dielectric structures that are software-manageable, allowing the antenna to dynamically modify its shape, position, and configuration to optimize radiation characteristics \cite{new2024_tutorial,wong2023_part1,new2023}. This unique feature, which sets FAS apart from traditional antenna systems (TASs), is made possible by advancements in liquid-metal technologies and RF-switchable pixels. 
The integration of FASs into WPCNs presents a transformative opportunity to enhance the adaptability and efficiency of these networks \cite{ghadi2024}. 
These innovations allow FAS to outperform TAS by enhancing energy efficiency and improving the reliability of data transmission in WPCNs \cite{lin2024,ghadi2024_backscatter}.
In the context of WPCNs, FAS enables real-time adaptation to varying channel conditions, ensuring optimal energy harvesting and information transfer in challenging environments, such as those with high mobility or severe multipath fading. By dynamically selecting the optimal antenna configuration based on channel state information (CSI), FAS minimizes energy wastage, maximizes transmission reliability, and ensures uninterrupted operation \cite{ye2023}. This capability is perfectly in line with the energy efficient and sustainable design principles of WPCNs, making FAS particularly valuable for applications in IoT, smart city infrastructures, and remote monitoring systems. 

\subsection{Related Works}
Physical layer security (PLS) has gained significant attention in WPCNs due to its potential to protect sensitive information. Several works have explored challenges and solutions in this domain.
Jiang et al. \cite{jiang2016} analyzed the secrecy performance of wirelessly powered wiretap channels, where an energy-constrained multi-antenna source, powered by a dedicated power beacon (PB), communicates in the presence of a passive eavesdropper. They proposed time-switching protocols and evaluated maximum ratio transmission and transmit antenna selection schemes. Their results demonstrated the critical role of full CSI in achieving substantial secrecy diversity gain and highlighted that wireless power transfer randomness does not degrade secrecy performance in high signal-to-noise ratio (SNR) regimes.
Huang et al. \cite{huang2018} studied the secrecy performance in wireless communication networks powered by multiple eavesdroppers. They considered the maximum ratio transmission and transmit antenna selection schemes under both non-colluding and colluding eavesdroppers. Closed-form expressions for secrecy outage probability (SOP) and average secrecy capacity (ASC) were derived, showing that feedback delay significantly impacts secrecy diversity. Their findings indicate that the transmit antenna selection scheme performs comparably or even better than the maximum ratio transmission under moderate feedback delays.
Cao et al. \cite{cao2023} explored PLS in reconfigurable intelligent surfaces (RIS)-aided WPCNs for IoT. They proposed three RIS-aided secure WPC modes, optimizing IRS phase shifts to enhance energy and information transfer. Their analysis revealed that mode-I achieves the best connection outage probability (COP), mode-II the best SOP, and mode-II outperforms others in ASC with higher transmission power or RIS elements, demonstrating the critical role of IRS configurations in balancing reliability and security.

FASs have emerged as transformative technologies for improving PLS, with several studies exploring its potential in enhancing secure communication under diverse scenarios.
The authors in \cite{ghadi2024_fas_pls} investigated the PLS performance of FAS-aided communication systems over arbitrary correlated fading channels. They derived analytical expressions for secrecy metrics, such as ASC, SOP, and secrecy energy efficiency (SEE), using Gaussian copulas and Gauss-Laguerre quadrature. Their findings demonstrated that FAS with a single activated port outperforms TASs, such as maximal ratio combining (MRC) and antenna selection (AS), in secure transmission reliability.
Vega-Sánchez et al. \cite{vega2024} analyzed the SOP of FAS in systems experiencing spatially correlated Nakagami-\(m\) fading. They compared non-diversity FAS and maximum-gain combining FAS (MGC-FAS) against multiple antennas with MRC at the eavesdropper. Their closed-form approximations revealed that MGC-FAS achieves superior secrecy performance when the FAS tube's size is sufficiently large and the legitimate path's average SNR significantly exceeds the eavesdropper's link SNR.
The authors in \cite{ghadi2024_ris} explored the integration of RIS with FAS for secure communications in wiretap scenarios. Their analysis focused on deriving the cumulative distribution function (CDF) and probability density function (PDF) of the SNR at each node using Gaussian copulas, leading to compact expressions for the SOP. The results demonstrated that the combination of FAS with RIS substantially enhances secure communication performance, underscoring the synergy between these technologies in next-generation networks.

Most existing studies on FAS integration have predominantly focused on the use of orthogonal multiple access (OMA) schemes. Although OMA has been widely adopted because of its simplicity, it is not the most efficient approach when the CSI at the transmitter is known. In this regard, several studies have demonstrated the significant advantages of the non-orthogonal multiple access (NOMA) principle over OMA in FAS \cite{new2023_noma,tlebaldiyeva2023,zheng2024}. Using superposition coding and successive interference cancellation (SIC), NOMA allows multiple users to share the same frequency resources, improving spectral efficiency and system performance. Therefore, highlighting the advantages of NOMA in FAS, the authors in \cite{ghadi2024} integrated WPCN with FAS under a NOMA scenario, demonstrating how FAS outperforms TAS for both near and far users in the WPCN.
\subsection{Motivations and Contributions}

While various combinations of FAS, WPCNs, NOMA, and PLS have been individually studied in prior research, their simultaneous integration for secure communication remains unexplored. The potential synergy of combining FAS with NOMA-enabled WPCNs to enhance PLS presents a promising yet unresolved research question. As these systems become increasingly critical to next-generation networks, ensuring their security is of utmost importance. However, their unique structure introduces distinct vulnerabilities that demand dedicated attention. Therefore, building on the work in \cite{ghadi2024}, we propose extending the system to a secure communication framework by combining FAS with NOMA-enabled WPCNs. This scenario potentially involves two eavesdropping models: (i) an external eavesdropper positioned between legitimate users, attempting to intercept transmissions, and (ii) a far NOMA user acting as an internal eavesdropper, exploiting shared resources to compromise communication secrecy. In this regard, we analyze the secrecy performance of both the external and internal eavesdropping scenarios, deriving a compact analytical framework of key secrecy metrics such as SOP and ASC. The main contributions of this paper are summarized as follows:

\begin{itemize}
    \item We propose a novel FAS-aided WPCN framework that integrates energy efficiency with secure transmission under the NOMA scheme, introducing two eavesdropping scenarios: external and internal.
    
    \item For both eavesdropping scenarios, we first derive the CDF and PDF of the equivalent channel gain for each FAS-equipped NOMA user in terms of the multivariate normal distribution. Subsequently, we obtain compact analytical expressions for the SOP and ASC of both users under external and internal eavesdropping scenarios, employing Gaussian quadrature techniques.
        
    \item We validate our analytical findings through Monte Carlo simulations, which confirm their accuracy and demonstrate the superior performance of the proposed system. The results show that deploying FAS instead of TAS in the WPCN significantly enhances transmission reliability and security. Specifically, FAS-equipped NOMA users achieve substantially lower SOP and higher ASC compared to TAS-equipped users, under both external and internal eavesdropping scenarios.
\end{itemize}
\subsection{Organization}
The rest of this paper is organized as follows. Section \ref{sec-sys} presents the system model, including the channel model and the NOMA-FAS scheme. Section \ref{sec-secrecy} provides the secrecy performance analysis, where theoretical frameworks for SOP and ASC are derived. Section \ref{sec-num} discusses the numerical results that validate the analytical findings and finally, Section \ref{sec-con} concludes the paper. 
\subsection{Notation}
We use boldface upper and lower case letters for matrices
and vectors, e.g. $\mathbf{X}$ and $\mathbf{x}$, respectively. 
Moreover, $(\cdot)^T$ , $(\cdot)^{-1}$, $|\cdot|$, $\mathrm{max}\{\cdot\}$, and $\mathrm{det}(\cdot)$ stand for the transpose, inverse, magnitude, maximum, and determinant operators, respectively.

	\section{System Model}\label{sec-sys}
This section presents the system model, where a single-antenna transmitter communicates confidential information to multiple FAS-equipped users under a NOMA scheme. The description includes key assumptions, such as channel conditions and FAS adaptability, which are critical for analyzing secrecy metrics in the following sections.
    
    \subsection{Channel Model} 

We consider a PLS scenario in a WPCN, as illustrated in Fig. \ref{fig-system}. 
In this setup, a transmitter $\mathrm{t}$ aims to deliver confidential information to two legitimate users, $\mathrm{u_n}$ and $\mathrm{u_f}$, over wireless fading channels, while an external eavesdropper $\mathrm{e}$ attempts to intercept the transmitted signal. The transmitter broadcasts a superimposed signal containing the confidential message intended for the users under the NOMA scheme, utilizing energy wirelessly supplied by a power beacon (PB). In this context, the wireless channel between the PB and the transmitter $\mathrm{t}$ is referred to as the \textit{energy link}. The channels connecting the transmitter $\mathrm{t}$ to the near user $\mathrm{u_n}$ (referred to as the strong user) and the far user $\mathrm{u_f}$ (referred to as the weak user) are termed \textit{information links}. Additionally, the channel between the transmitter $\mathrm{t}$ and the external eavesdropper $\mathrm{e}$ is designated as the \textit{eavesdropper link}. Moreover, it is assumed that the far user $\mathrm{u_f}$ may act as an internal eavesdropper, attempting to decode the data signal intended for the near user $\mathrm{u_n}$. It is assumed that both the PB and the transmitter $\mathrm{t}$ are single-antenna nodes, whereas the NOMA users $i\in{\mathrm{u_n}, \mathrm{u_f}}$ and the eavesdropper $\mathrm{e}$ are equipped with planar FASs. Specifically, the nodes $j\in\left\{\mathrm{u_n,u_f,e}\right\}$ incorporate a grid structure comprising $N^l_j$ ports, which are uniformly distributed along a linear space of length $W^l_j\lambda$ for $l\in\left\{1,2\right\}$. Consequently, the total number of ports is given by $N_j = N^1_j \times N^2_j$, and the corresponding planar dimensions are $W_j = W^1_j\lambda \times W^2_j\lambda$. Additionally, to facilitate the transformation between the 2D and 1D index representations, a mapping function $\mathcal{F}\left(n_j\right) = \left(n_j^1, n_j^2\right)$ and its inverse $\mathcal{F}^{-1}\left(n_j^1, n_j^2\right) = n_j$ are defined, where $n_j\in\left\{1,\dots,N_j\right\}$ and $n_j^l\in\left\{1,\dots,N_j^l\right\}$.

The received signal at the $n_j$-th port of the node $j$ is expressed as
		\begin{align}
		y_j^{n_j}=\sqrt{P_\mathrm{p}\delta_j}h_{\mathrm{eq},j}^{n_j}\left(\sqrt{p_\mathrm{u_n}}x_\mathrm{u_n}+\sqrt{p_\mathrm{u_f}}x_\mathrm{u_f}\right)+z_j^{n_j},
	\end{align}
	in which $P_\mathrm{p}$ represents the power transmitted by the PB, and $h_{\mathrm{eq},j}^{n_j} = h_\mathrm{t} h_j^{n_j}$ denotes the equivalent channel coefficient at the $n_j$-th port of node $j$. This coefficient encompasses the channel coefficients of the energy link $h_\mathrm{t}$, the information links $h_i^{n_i}$, and the eavesdropper link $h_\mathrm{e}^{n_\mathrm{e}}$. Additionally, $x_i$ represents the symbol transmitted to user $i$ with unit power, while $p_i$ specifies the power allocation factor for user $i$, such that $p_\mathrm{u_1} + p_\mathrm{u_2} = 1$. The term $z_j^{n_j}$ denotes the additive white Gaussian noise (AWGN) at the $n_j$-th port of node $j$, characterized by zero mean and variance $\sigma_j^2$. Moreover, $\delta_j = L_\mathrm{p} d_\mathrm{t}^{-\alpha} d_j^{-\alpha}$ describes the path-loss, where $L_\mathrm{p}$ represents the frequency-dependent signal propagation loss, $\alpha > 2$ denotes the path-loss exponent, $d_\mathrm{t}$ is the distance between the PB and the transmitter $\mathrm{t}$, and $d_j$ is the distance between the transmitter $\mathrm{t}$ and node $j$. 
	\begin{figure}[!t]
	 	\centering	 \includegraphics[width=0.9998\columnwidth]{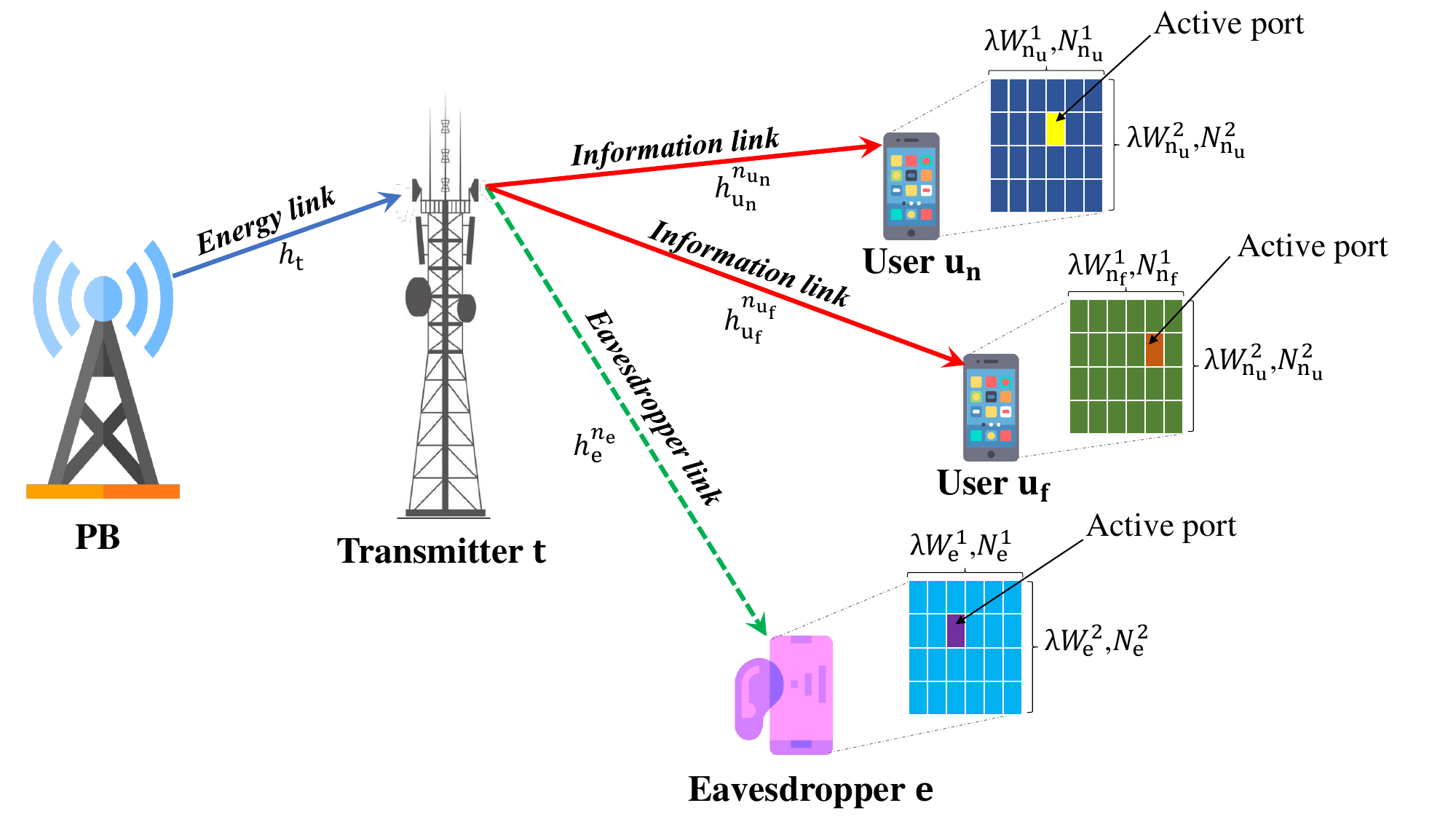}\vspace{0cm}
	 	\caption{System model: FAS-aided secure NOMA WPCN.}\vspace{-0.5cm}
	 	\label{fig-system}
	 \end{figure}
     
	Furthermore, the ports in the FAS can freely switch and are positioned in close proximity to each other, leading to spatial correlation in the channel coefficients at each NOMA FAS-equipped node. Assuming that the energy, information, and eavesdropper links experience Rayleigh fading, the covariance between two arbitrary ports $n_j$ and $\tilde{n}_j$ at each node $j$ in a three-dimensional (3D) rich-scattering environment can be expressed as
	\begin{align}
	\varrho_{n_j,\tilde{n}_j}^j=\mathcal{J}_0\left(2\pi\sqrt{\left(\frac{n^1_i-\tilde{n}^1_j}{N_j^1-1}W_j^1\right)^2+\left(\frac{n^2_j-\tilde{n}^2_j}{N_j^2-1}W_j^2\right)^2}\right),
	\end{align}
	 in which $\tilde{n}_j=\mathcal{F}^{-1}\left(\tilde{n}_j^1,\tilde{n}_j^2\right)$ so that $\tilde{n}_j^l\in\left\{1,\dots,N_j^l\right\}$ and $\mathcal{J}_0(.)$ denotes  the zero-order spherical Bessel function of the first kind.
	 
	 Given the concept of FAS, we assume that only the optimal port, which maximizes the received signal-to-noise ratio (SNR) at node $j$, is activated. Additionally, we consider a scenario where the channel state information (CSI) is available to the external eavesdropper, enabling it to select the best port to maximize its SNR. This assumption represents a worst-case scenario, as the eavesdropper is provided with ideal operating conditions. Specifically, the external eavesdropper is assumed to perform perfect multi-user detection, allowing it to differentiate between user data and completely eliminate inter-user interference. Moreover, due to the inherent spatial correlation among fluid antenna ports, acquiring full CSI requires only a limited number of port observations, irrespective of the total number of fluid antenna ports \cite{noor2023deep}. Consequently, the equivalent channel gain at node $j$ is defined as
	 		\begin{align}
	 	g_{\mathrm{fas},j}=\max\left\{g_{\mathrm{eq},1}^1,\dots, g_{\mathrm{eq},j}^{N_j}\right\} = g_{\mathrm{eq},j}^{n_{j}^*}, \label{eq-g-fas}
	 \end{align}
	 in which $g_{\mathrm{eq},j}^{n_{j}^*} = \left|h_{\mathrm{eq},j}^{n_j}\right|^2$ is the channel gain at $n_j$-th port of node $j$ and $n_j^*$ denotes the best port index at node $j$ that maximizes the channel gain, i.e.,
	 \begin{align}
	 	n_u^* = \arg\underset{n_u}{\max}\left\{\left|h_{\mathrm{eq},j}^{n_j}\right|^2\right\}.
	 \end{align}
	 
	 \subsection{NOMA-FAS Scheme}
In accordance with the NOMA principle, the near FAS-equipped user $\mathrm{u_n}$ employs successive interference cancellation (SIC) as part of its optimal decoding strategy. This approach requires the far FAS-equipped user $\mathrm{u_f}$ to decode its signal directly while treating any interference as background noise. Assuming that only the optimal port, which maximizes the received signal-to-interference-plus-noise ratio (SINR) for the FAS-equipped users, is active, the SINR for the SIC process is formulated as
\vspace{-6pt}
	\begin{align}
	\gamma_\mathrm{sic}=\frac{\overline{\gamma}_\mathrm{u_n}p_\mathrm{u_f}g_{\mathrm{fas},{\mathrm{u_n}}}}{\overline{\gamma}_\mathrm{u_n}p_\mathrm{u_n}g_{\mathrm{fas},{\mathrm{u_n}}}+1},\label{eq-sinr-sic}
	\end{align}	 
	where $\overline{\gamma}_\mathrm{u_n}=\frac{P_\mathrm{p}L_\mathrm{p}}{\sigma^2_\mathrm{u_n}d_\mathrm{t}^{\alpha}d_\mathrm{u_n}^{\alpha}}$ denotes the average SNR for the near user $\mathrm{u_n}$. After applying SIC, $\mathrm{u_n}$ eliminates the signal intended for $\mathrm{u_f}$ from its received signal and proceeds to decode its own required information. Consequently, the received SNR at the FAS-equipped user $\mathrm{u_n}$ is expressed as 
	\begin{align}
		\gamma_\mathrm{u_n}=\overline{\gamma}_\mathrm{u_n}p_\mathrm{u_n}g_{\mathrm{fas},{\mathrm{u_n}}}. \label{eq-snr-u1}
	\end{align}
Meanwhile, the far user $\mathrm{u_f}$ decodes its own signal directly but is unable to distinguish or cancel out the signal from the near user $\mathrm{u_n}$ within the superimposed transmission. Consequently, the SINR experienced by the far FAS-equipped user $\mathrm{u_f}$ is represented as
	\begin{align}
		\gamma_\mathrm{u_f}=\frac{\overline{\gamma}_\mathrm{u_f}p_\mathrm{u_f}g_{\mathrm{fas},{\mathrm{u_f}}}}{\overline{\gamma}_\mathrm{u_f}p_\mathrm{u_n}g_{\mathrm{fas},{\mathrm{u_f}}}+1},\label{eq-sinr-uf}
	\end{align}
where $\overline{\gamma}_\mathrm{u_f}=\frac{P_\mathrm{p}L_\mathrm{p}}{\sigma^2_\mathrm{u_f}d_\mathrm{t}^{\alpha}d_\mathrm{u_f}^{\alpha}}$ defines the average SNR for the near user.

The external eavesdropper intercepts the transmitted signals from the transmitter. As a result, the SNR at the external eavesdropper $\mathrm{e}$ for decoding $x_i$ is given by
\begin{align}
\gamma_{\mathrm{e},i}=\overline{\gamma}_\mathrm{e}p_ig_\mathrm{fas,e},\label{eq-sinr-e}
\end{align}
where $\overline{\gamma}_\mathrm{e}=\frac{P_\mathrm{p}L_\mathrm{p}}{\sigma^2_\mathrm{e}d_\mathrm{t}^{\alpha}d_\mathrm{e}^{\alpha}}$ represents the average SNR of the external eavesdropper.

In the case of internal eavesdropping, the far user attempts to decode the message intended for the near user after receiving $x_\mathrm{u_f}$. Hence, assuming that $x_\mathrm{u_f}$ is decoded without error at the far user, the SNR at the far user $\mathrm{u_f}$ for decoding $x_\mathrm{u_n}$ is expressed as
\begin{align}
\gamma_{\mathrm{u_f,n}} = \overline{\gamma}_\mathrm{u_f}p_\mathrm{u_n}g_\mathrm{fas,u_f}.\label{eq-sinr-ufn}
\end{align}

\section{Secrecy Performance Analysis} \label{sec-secrecy}
In this section, we first introduce the distribution of the equivalent channel for the FAS-equipped nodes. Using this distribution, we then derive analytical expressions for the SOP and the ASC, which provide insights into the secrecy performance of the system under various conditions. 
\subsection{SOP Analysis}
\subsubsection{Statistical Characteristics}
To analytically assess key secrecy performance metrics in the considered FAS-aided secure WPCN, the first step is to derive the distribution of the equivalent fading channel coefficients at the FAS-equipped nodes. Specifically, this involves calculating the maximum of $N_j$ correlated random variables (RVs), each of which is the product of two independent exponentially-distributed RVs (as shown in \eqref{eq-g-fas}). Since the energy, information, and eavesdropper links experience Rayleigh fading, the corresponding channel gains $g_\mathrm{t} = \left|h_\mathrm{t}\right|^2$ and $g_j^{n_j} = \left|h_j^{n_j}\right|^2$ follow an exponential distribution with unit mean. Consequently, the cumulative distribution function (CDF) of the equivalent channel gain $g_{\mathrm{eq},j}^{n_j} = g_\mathrm{t} g_j^{n_j}$ at node $j$ is derived as
\begin{align}
	F_{g_{\mathrm{eq},j}^{n_j}}(g)&=\int_0^\infty f_{g_\mathrm{t}}(g_\mathrm{t})F_{g_j^{n_j}}\left(\frac{g}{g_\mathrm{t}}\right)\mathrm{d}g_\mathrm{t}\\
	&=1-\int_0^\infty \exp\left(-\left(g_\mathrm{t}+\frac{g}{g_\mathrm{t}}\right)\right)\mathrm{d}g_\mathrm{t}\\
	& \overset{(a)}{=}1-2\sqrt{g}\mathcal{K}_1\left(2\sqrt{g}\right), \label{eq-cdf-eq}
\end{align}
in which $f_{g_\mathrm{t}}(\cdot)$ and $F_{g_i^{n_i}}(\cdot)$ are the marginal PDF and CDF of exponentially-distributed $g_\mathrm{t}$ and $g_j^{n_j}$. Additionally, $(a)$ is derived by employing the integral expression given in \cite[3.471.9]{gradshteyn2007table}, where $\mathcal{K}_1(\cdot)$ denotes the first-order modified Bessel function of the second kind.

Next, by applying the definition, the CDF of $g_{\mathrm{fas},j}$ can be mathematically formulated as follows
\begin{align}
	F_{g_{\mathrm{fas},j}}(g)&=\Pr\left(\max\left\{g_{\mathrm{eq},j}^1,\dots,g_{\mathrm{eq},j}^{N_j}\right\}\leq g\right)\\
	&=\Pr\left(g_{\mathrm{eq},j}^1\leq r,\dots,g_{\mathrm{eq},j}^{N_j}\leq g\right)\\
	&=F_{g_{\mathrm{eq},j}^1,\dots,g_{\mathrm{eq},j}^{N_j}}\left(g,\dots,g\right), \label{eq-fas1}
\end{align}
where  \eqref{eq-fas1} represents the definition of the multivariate distribution for $N_j$ correlated RVs $g_{\mathrm{eq},j}^{n_j}$, where $n_j \in \left\{1, \dots, N_j\right\}$. To derive this distribution, we apply Sklar's theorem \cite{Nelsen2006intro}, which allows us to construct the multivariate distribution of $N_j$ correlated arbitrary RVs, such as $g_{\mathrm{fas},j}$, by utilizing only the corresponding marginal distributions. Therefore, the CDF of $g_{\mathrm{fas},j}$, denoted as $F_{g_{\mathrm{fas},j}}(g)$, can be expressed as
\begin{align}\notag
	&F_{g_{\mathrm{fas},j}}\left(g\right)=\\
	&\quad\Phi_{\vec{R}_j}\left(\phi^{-1}\left(F_{g_{\mathrm{eq},j}^1}\left(g\right)\right)\hspace{-1mm},\dots,\phi^{-1}\left(F_{g_{\mathrm{eq},j}^{N_j}}\left(g\right)\right)\hspace{-1mm};\vartheta_j\right),\label{eq-cdf-fas}
\end{align}
where $\phi^{-1}\left(F_{g_{\mathrm{eq},j}^{n_j}}\left(g\right)\right)=\sqrt 2\mathrm{erf}^{-1}\left(2F_{g_{\mathrm{eq},j}^{n_j}}\left(g\right)-1\right)$ represents the quantile function of the standard normal distribution, and $\vartheta_j$ denotes the dependence parameter, specifically the Gaussian copula parameter. \footnote{The Gaussian copula effectively captures the spatial correlation between fluid antenna ports, particularly when the fluid antenna size is large. In this context, the dependence parameter of the Gaussian copula can be approximated by the covariance between two arbitrary ports, i.e., $\varrho^i \approx \vartheta_i$ \cite{ghadi2024gau}.} The term $\Phi_{\vec{R}_j}$ refers to the joint CDF of the multivariate normal distribution, which has a zero mean vector and a correlation matrix $\vec{R}_j$.
\begin{align}
	\vec{R}_j=\begin{bmatrix}
		\varrho_{1,1}^j & \varrho_{1,2}^j &\dots& \varrho_{1,N_j}^j\\
		\varrho_{2,1}^j & \varrho_{2,2}^j &\dots& \varrho_{2,N_j}^j\\ \vdots & \vdots & \ddots & \vdots\\
		\varrho_{N_j,1}^j & \varrho_{N_j,2}^j &\dots& \varrho_{N_j,N_j}^j
	\end{bmatrix}.
\end{align}
By applying the chain rule to \eqref{eq-cdf-fas}, the PDF of $g_{\mathrm{fas},j}$ is expressed as
\begin{align}\notag
f_{g_{\mathrm{fas},j}}\left(g\right)&=\prod_{n_j=1}^{N_j}f_{g_{\mathrm{eq},j}^{n_j}}\left(g\right)\\
&\quad\times\frac{\exp\left(-\frac{1}{2}\left(\boldsymbol{\phi}^{-1}_{N_j}\right)^T\left(\vec{R}^{-1}_{j}-\vec{I}\right)\boldsymbol{\phi}^{-1}_{N_j}\right)}{\sqrt{{\rm det}\left(\vec{R}_{j}\right)}},\label{eq-pdf-fas}
\end{align}
where ${\rm det}\left(\vec{R}_{j}\right)$ denotes the determinant of the correlation matrix $\vec{R}_{j}$, $\vec{I}$ is the identity matrix, and 
\begin{align}
\boldsymbol{\phi}^{-1}_{N_j}=\left[\phi^{-1}\left(F_{g_{\mathrm{eq},j}^{1}}(g)\right),\dots,\phi^{-1}\left(F_{g_{\mathrm{eq},j}^{N_j}}(g)\right)\right]^T.
\end{align}
Moreover, $f_{g_{\mathrm{eq},j}^{n_j}}\left(g\right)$ is the PDF of the equivalent channel gain $g_{\mathrm{eq},j}^{n_j}=g_\mathrm{t}g_j^{n_j}$ at the node $j$, which is derived as
\begin{align}\label{eq-pdf-eq}
f_{g_{\mathrm{eq},j}^{n_j}}(g) = 2\mathcal{K}_0\left(\sqrt{2g}\right),
\end{align}
where $\mathcal{K}_0$ denotes the zero-order modified Bessel function of the second kind.

Therefore, by inserting \eqref{eq-cdf-eq} into \eqref{eq-cdf-fas} and replacing  \eqref{eq-pdf-eq} into \eqref{eq-pdf-fas}, the CDF and PDF of $g_{\mathrm{fas},j}\left(g\right)$ are derived as shown in \eqref{eq-cdf} and \eqref{eq-pdf}, respectively, where $\phi^{-1}_{N_j}$ is defined in \eqref{eq-phi}.
\begin{figure*}
	\begin{align}\label{eq-cdf}
		F_{g_{\mathrm{fas},j}}\left(g\right)= \Phi_{\vec{R}_j}\left(\sqrt 2\mathrm{erf}^{-1}\left(1-4\sqrt{g}\mathcal{K}_1\left(2\sqrt{g}\right)\right)\hspace{-1mm},\dots,\sqrt 2\mathrm{erf}^{-1}\left(1-4\sqrt{g}\mathcal{K}_1\left(2\sqrt{g}\right)\right);\vartheta_j\right)
	\end{align}
	\hrulefill
	\begin{align}
		f_{g_{\mathrm{fas},j}}\left(g\right)=\left(2\mathcal{K}_0\left(\sqrt{2g}\right)\right)^{N_j}\frac{\exp\left(-\frac{1}{2}\left(\boldsymbol{\phi}^{-1}_{N_j}\right)^T\left(\vec{R}^{-1}_{j}-\vec{I}\right)\boldsymbol{\phi}^{-1}_{N_j}\right)}{\sqrt{{\rm det}\left(\vec{R}_{j}\right)}},\label{eq-pdf}
	\end{align}
	\begin{align}\label{eq-phi}
	\boldsymbol{\phi}^{-1}_{N_j}=\left[\sqrt 2\mathrm{erf}^{-1}\left(1-4\sqrt{g}\mathcal{K}_1\left(2\sqrt{g}\right)\right),\dots,\sqrt 2\mathrm{erf}^{-1}\left(1-4\sqrt{g}\mathcal{K}_1\left(2\sqrt{g}\right)\right)\right]^T
	\end{align}
		\hrulefill
\end{figure*}
\subsubsection{External Eavesdropping}
Here, we derive the analytical expressions of the SOP for both the near and far users under the scenario of external eavesdropping.

The SOP is defined as the likelihood that the instantaneous secrecy capacity $\mathcal{C}_{\mathrm{s},i}^\mathrm{ext}$ in the external eavesdropping scenario falls below a specified target secrecy rate $R_{\mathrm{s},i}\ge 0$, which can be expressed as
\begin{align}
P_{\mathrm{sop},i}^\mathrm{ext} = \Pr\left(\mathcal{C}_{\mathrm{s},i}^\mathrm{ext}\leq R_{\mathrm{s},i}\right),\label{eq-sop-def1}
	\end{align}
where $\mathcal{C}_{\mathrm{s},i}^\mathrm{ext}$ 
is defined as the difference between the channel capacities of the legitimate users and the eavesdropper links, i.e.,
\begin{align}
\mathcal{C}_{\mathrm{s},i}^\mathrm{ext} = \max\left\{\log_2\left(1+\gamma_{i}\right)-\log_2\left(1+\gamma_{\mathrm{e},i}\right),0\right\}.\label{eq-sc-def1}
\end{align}
\begin{theorem}
The SOP expressions of the near user $\mathrm{u_n}$ and the far user $\mathrm{u_n}$ in the proposed FAS-aided secure WPCN  under the external eavesdropping scenario are given by \eqref{eq-sop-n} and \eqref{eq-sop-f}, respectively, where $\nu_m = \frac{\overline{R}_\mathrm{n}\left(\overline{\gamma}_\mathrm{e}p_\mathrm{u_n}\epsilon_m+1\right)-1}{\overline{\gamma}_\mathrm{u_n}p_\mathrm{u_n}}$, $\mu_m = \frac{\overline{R}_\mathrm{f}\left(1+\overline{\gamma}_\mathrm{e}p_\mathrm{u_f}\epsilon_m\right)-1}{\overline{\gamma}_\mathrm{u_f}\left(p_\mathrm{u_f}-p_\mathrm{u_n}\left[\overline{R}_\mathrm{f}\left(1+\overline{\gamma}_\mathrm{e}p_\mathrm{u_f}\epsilon_m\right)-1\right]\right)}$, and  $\boldsymbol{\phi}^{-1}_{N_\mathrm{e}}=\left[\phi^{-1}\left(F_{g_{\mathrm{eq},\mathrm{e}}^{1}}(\epsilon_m)\right),\dots,\phi^{-1}\left(F_{g_{\mathrm{eq},\mathrm{e}}^{N_\mathrm{e}}}(\epsilon_m)\right)\right]^T$.
\begin{figure*}
\begin{align}\notag
P_{\mathrm{sop,u_n}}^\mathrm{ext}\approx & \sum_{m=0}^{M}\mathrm{e}^{\epsilon_m}\omega_m \left(2\mathcal{K}_0\left(\sqrt{2\epsilon_m}\right)\right)^{N_\mathrm{e}}\frac{\exp\left(-\frac{1}{2}\left(\boldsymbol{\phi}^{-1}_{N_\mathrm{e}}\right)^T\left(\vec{R}^{-1}_\mathrm{e}-\vec{I}\right)\boldsymbol{\phi}^{-1}_{N_\mathrm{e}}\right)}{\sqrt{{\rm det}\left(\vec{R}_\mathrm{e}\right)}}\\
&\times
\Phi_{\vec{R}_\mathrm{u_n}}\left(\sqrt 2\mathrm{erf}^{-1}\left(1-4\sqrt{\nu_m}\mathcal{K}_1\left(2\sqrt{\nu_m}\right)\right)\hspace{-1mm},\dots,\sqrt 2\mathrm{erf}^{-1}\left(1-4\sqrt{\nu_m}\mathcal{K}_1\left(2\sqrt{\nu_m}\right)\right);\vartheta_\mathrm{u_n}\right)\label{eq-sop-n}
\end{align}
\hrulefill
	\begin{align}\notag
	P_{\mathrm{sop,u_f}}^\mathrm{ext}\approx & \sum_{m=0}^{M}\mathrm{e}^{\epsilon_m}\omega_m \left(2\mathcal{K}_0\left(\sqrt{2\epsilon_m}\right)\right)^{N_\mathrm{e}}\frac{\exp\left(-\frac{1}{2}\left(\boldsymbol{\phi}^{-1}_{N_\mathrm{e}}\right)^T\left(\vec{R}^{-1}_\mathrm{e}-\vec{I}\right)\boldsymbol{\phi}^{-1}_{N_\mathrm{e}}\right)}{\sqrt{{\rm det}\left(\vec{R}_\mathrm{e}\right)}}\\
	&\times
	\Phi_{\vec{R}_\mathrm{u_f}}\left(\sqrt 2\mathrm{erf}^{-1}\left(1-4\sqrt{\mu_m}\mathcal{K}_1\left(2\sqrt{\mu_m}\right)\right)\hspace{-1mm},\dots,\sqrt 2\mathrm{erf}^{-1}\left(1-4\sqrt{\mu_m}\mathcal{K}_1\left(2\sqrt{\mu_m}\right)\right);\vartheta_\mathrm{u_f}\right)\label{eq-sop-f}
\end{align}
\hrulefill
\end{figure*}
\end{theorem}
\begin{proof}
By substituting \eqref{eq-snr-u1} and \eqref{eq-sinr-e} into \eqref{eq-sop-def1}, we have
\begin{align}
&P_{\mathrm{sop,u_n}}^\mathrm{ext} = \Pr\left(\frac{1+\gamma_\mathrm{u_n}}{1+\gamma_\mathrm{e,u_n}}\leq 2^{R_{\mathrm{s,u_n}}}\right)\\
& = \Pr\left(\frac{1+\overline{\gamma}_\mathrm{u_n}p_\mathrm{u_n}g_{\mathrm{fas},{\mathrm{u_n}}}}{1+\overline{\gamma}_\mathrm{e}p_\mathrm{u_n}g_\mathrm{fas,e}}\leq \overline{R}_\mathrm{n}\right)\\
&=\Pr\left(g_\mathrm{fas,u_n}\leq\frac{\overline{R}_\mathrm{n}\left(\overline{\gamma}_\mathrm{e}p_\mathrm{u_n}g_\mathrm{fas,e}+1\right)-1}{\overline{\gamma}_\mathrm{u_n}p_\mathrm{u_n}}\right)\\
&=\int_0^\infty F_{g_\mathrm{fas,u_n}}\left(\frac{\overline{R}_\mathrm{n}\left(\overline{\gamma}_\mathrm{e}p_\mathrm{u_n}g+1\right)-1}{\overline{\gamma}_\mathrm{u_n}p_\mathrm{u_n}}\right)f_{g_\mathrm{fas,e}}(g)\mathrm{d}g,\label{eq-p1}
\end{align}
where $\overline{R}_\mathrm{n}=2^{R_{\mathrm{s,u_n}}}$. After substituting the corresponding CDF and PDF into \eqref{eq-p1}, it becomes evident that solving the integral analytically is intractable. Therefore, we apply the Gauss–Laguerre quadrature method, as outlined in the following lemma, to derive the SOP.
\begin{lemma}\label{lemma-glq}
The Gauss–Laguerre quadrature is defined as \cite{Abromowitz1972handbook}
\begin{align}
	\int_0^\infty \Lambda\left(x\right)\mathrm{d}x\approx  \sum_{m=0}^{M}\mathrm{e}^{\epsilon_m}\omega_m\Lambda\left(\epsilon_m\right),
\end{align}
where $M$ is the number of sample points used, $\epsilon_m$ denotes the $m$-th root of Laguerre polynomial $L_M\left(\epsilon_m\right)$, and $\omega_m=\frac{\epsilon_m}{2\left(M+1\right)^2L_{M+1}^2\left(\epsilon_m\right)}$ represents the corresponding weight. 
\end{lemma}
Now, by applying Lemma \ref{lemma-glq} to \eqref{eq-p1}, the expression for \eqref{eq-sop-n} is obtained, thus completing the proof. Similarly, or the SOP of the far user, we have
\begin{align}
&P_{\mathrm{sop,u_f}}^\mathrm{ext} = \Pr\left(\frac{1+\gamma_\mathrm{u_f}}{1+\gamma_\mathrm{e,u_f}}\leq 2^{R_{\mathrm{s,u_f}}}\right)\\
& = \Pr\left(\frac{1+\frac{\overline{\gamma}_\mathrm{u_f}p_\mathrm{u_f}g_{\mathrm{fas},{\mathrm{u_f}}}}{\overline{\gamma}_\mathrm{u_f}p_\mathrm{u_n}g_{\mathrm{fas},{\mathrm{u_f}}}+1}}{1+\overline{\gamma}_\mathrm{e}p_\mathrm{u_f}g_\mathrm{fas,e}}\leq \overline{R}_\mathrm{f}\right)\\
& = \Pr\hspace{-1mm}\left(\hspace{-1mm}g_\mathrm{fas,u_f}\hspace{-1mm}\leq \frac{\overline{R}_\mathrm{f}\left(1+\overline{\gamma}_\mathrm{e}p_\mathrm{u_f}g_\mathrm{fas,e}\right)-1}{\overline{\gamma}_\mathrm{u_f}\left(p_\mathrm{u_f}-p_\mathrm{u_n}\left[\overline{R}_\mathrm{f}\left(1+\overline{\gamma}_\mathrm{e}p_\mathrm{u_f}g_\mathrm{fas,e}\right)-1\right]\right)}\right)\\\notag
&=\int_0^\infty F_{g_\mathrm{fas,u_f}}\left(\frac{\overline{R}_\mathrm{f}\left(1+\overline{\gamma}_\mathrm{e}p_\mathrm{u_f}g\right)-1}{\overline{\gamma}_\mathrm{u_f}\left(p_\mathrm{u_f}-p_\mathrm{u_n}\left[\overline{R}_\mathrm{f}\left(1+\overline{\gamma}_\mathrm{e}p_\mathrm{u_f}g\right)-1\right]\right)}\right)\\
&\quad\quad\times f_{g_\mathrm{fas,e}}(g)\mathrm{d}g,\label{eq-p2}
\end{align}
where $\overline{R}_\mathrm{f}=2^{R_{\mathrm{s,u_f}}}$. After substituting the corresponding CDF and PDF into \eqref{eq-p2}, and applying Lemma \ref{lemma-glq}, the expression for \eqref{eq-sop-f} is derived, thus completing the proof.
\end{proof}
\subsubsection{Internal Eavesdropping}
Given the definition of the internal eavesdropping scenario, where the far user acts as an eavesdropper, the corresponding secrecy capacity is defined as
\begin{align}\label{eq-sc-def2}
	\mathcal{C}_{\mathrm{s,u_n}}^\mathrm{int} = \max\left\{\log_2\left(1+\gamma_\mathrm{u_n}\right)-\log_2\left(1+\gamma_{\mathrm{u_f,n}}\right),0\right\},
\end{align}
and the corresponding SOP is defined as 
\begin{align}
P_{\mathrm{sop,u_n}}^\mathrm{int} = \Pr\left(\mathcal{C}_{\mathrm{s,u_n}}^\mathrm{int}\leq R_{\mathrm{s,u_n}}\right). \label{eq-def-sop2}
\end{align}
\begin{theorem}
The SOP expression of the near user $\mathrm{u_n}$ in the proposed FAS-aided secure WPCN  under the internal eavesdropping scenario is given by \eqref{eq-sop-in}, where $\zeta_m=\frac{\overline{R}_\mathrm{n}\left(\overline{\gamma}_\mathrm{u_f}p_\mathrm{u_n}\epsilon_m+1\right)-1}{\overline{\gamma}_\mathrm{u_n}p_\mathrm{u_n}}$ and $\boldsymbol{\phi}^{-1}_{N_\mathrm{u_f}}=\left[\phi^{-1}\left(F_{g_{\mathrm{eq},\mathrm{u_f}}^{1}}(\epsilon_m)\right),\dots,\phi^{-1}\left(F_{g_{\mathrm{eq},\mathrm{u_f}}^{N_\mathrm{u_f}}}(\epsilon_m)\right)\right]^T$.
\begin{figure*}
	\begin{align}\notag
P_{\mathrm{sop,u_n}}^\mathrm{int}\approx & \sum_{m=0}^{M}\mathrm{e}^{\epsilon_m}\omega_m \left(2\mathcal{K}_0\left(\sqrt{2\epsilon_m}\right)\right)^{N_\mathrm{u_f}}\frac{\exp\left(-\frac{1}{2}\left(\boldsymbol{\phi}^{-1}_{N_\mathrm{u_f}}\right)^T\left(\vec{R}^{-1}_\mathrm{u_f}-\vec{I}\right)\boldsymbol{\phi}^{-1}_{N_\mathrm{u_f}}\right)}{\sqrt{{\rm det}\left(\vec{R}_\mathrm{u_f}\right)}}\\
&\times
\Phi_{\vec{R}_\mathrm{u_n}}\left(\sqrt 2\mathrm{erf}^{-1}\left(1-4\sqrt{\zeta_m}\mathcal{K}_1\left(2\sqrt{\zeta_m}\right)\right)\hspace{-1mm},\dots,\sqrt 2\mathrm{erf}^{-1}\left(1-4\sqrt{\zeta_m}\mathcal{K}_1\left(2\sqrt{\zeta_m}\right)\right);\vartheta_\mathrm{u_n}\right)\label{eq-sop-in}
\end{align}
\hrulefill
\end{figure*}
\end{theorem}
\begin{proof}
	As given by \eqref{eq-def-sop2}, we have
\begin{align}
	&P_{\mathrm{sop,u_n}}^\mathrm{int}= \Pr\left(\frac{1+\gamma_{\mathrm{u_n}}}{1+\gamma_\mathrm{u_f,n}}\leq 2^{R_{\mathrm{s,u_n}}}\right)\\
	&=\Pr\left(\frac{1+\overline{\gamma}_\mathrm{u_n}p_\mathrm{u_n}g_{\mathrm{fas},{\mathrm{u_n}}}}{1+\overline{\gamma}_\mathrm{u_f}p_\mathrm{u_n}g_\mathrm{fas,u_f}}\leq \overline{R}_\mathrm{n}\right)\\
&=\Pr\left(g_\mathrm{fas,u_n}\leq\frac{\overline{R}_\mathrm{n}\left(\overline{\gamma}_\mathrm{u_f}p_\mathrm{u_n}g_\mathrm{fas,u_f}+1\right)-1}{\overline{\gamma}_\mathrm{u_n}p_\mathrm{u_n}}\right)\\
&=\int_0^\infty F_{g_\mathrm{fas,u_n}}\left(\frac{\overline{R}_\mathrm{n}\left(\overline{\gamma}_\mathrm{u_f}p_\mathrm{u_n}g+1\right)-1}{\overline{\gamma}_\mathrm{u_n}p_\mathrm{u_n}}\right)f_{g_\mathrm{fas,u_f}}(g)\mathrm{d}g.\label{eq-p3}
\end{align}
Now, by applying Lemma \ref{lemma-glq} to \eqref{eq-p3}, \eqref{eq-sop-in} is derived and the proof is completed.  
\end{proof}

	\subsection{ASC Analysis}
	\subsubsection{External Eavesdropping}Given the definition of secrecy capacity in \eqref{eq-sc-def1}, the ASC under the external eavesdropping scenario is defined as
	\begin{align}
\overline{\mathcal{C}}_{\mathrm{s},i}^\mathrm{ext}=\int_0^\infty\int_0^{\gamma_{\mathrm{u_n}}}\mathcal{C}_{\mathrm{s},i}^\mathrm{ext}\left(\gamma_i,\gamma_{\mathrm{e},i}\right)f_{\gamma_i}\left(\gamma_i\right)f_{\gamma_{\mathrm{e},i}}\left(\gamma_{\mathrm{e},i}\right)\mathrm{d}\gamma_i\mathrm{d}\gamma_{\mathrm{e},i},\label{eq-asc-def1}
	\end{align}
where $f_{\gamma_i}\left(\gamma_i\right)$ and $f_{\gamma_{\mathrm{e},i}}\left(\gamma_{\mathrm{e},i}\right)$ are the PDFs of $\gamma_i$ and $\gamma_{\mathrm{e},i}$, respectively, which can be obtained through the transformation of RVs.  
	\begin{theorem}
	The ASC expressions of the near user $\mathrm{u_n}$ and the far user $\mathrm{u_n}$ in the proposed FAS-aided secure WPCN  under the external eavesdropping scenario are given by \eqref{eq-asc-un} and \eqref{eq-asc-uf}, respectively, where $\eta_m = \frac{\epsilon_m}{\overline{\gamma}_\mathrm{u_n}p_\mathrm{u_n}}$, $\tau_m = \frac{\epsilon_m}{\overline{\gamma}_\mathrm{e}p_\mathrm{u_f}}$, $\xi_{\tilde{m}}=\frac{\psi_{\tilde{m}}}{\overline{\gamma}_\mathrm{u_f}\left(p_\mathrm{u_f}-\psi_{\tilde{m}} p_\mathrm{u_n}\right)}$, and $\chi_{\tilde{m}}=\frac{\psi_{\tilde{m}}}{\overline{\gamma}_\mathrm{e}p_{\mathrm{u_f}}}$.
	\begin{figure*}
	\begin{align}\notag
		\overline{\mathcal{C}}_{\mathrm{s,u_n}}^\mathrm{ext} \approx &\frac{1}{\ln 2}\sum_{m=0}^M \frac{\mathrm{e}^{\epsilon_m}\omega_m}{1+\epsilon_m}\left[1-\Phi_{\vec{R}_\mathrm{u_n}}\left(\sqrt 2\mathrm{erf}^{-1}\left(1-4\sqrt{\eta_m}\mathcal{K}_1\left(2\sqrt{\eta_m}\right)\right)\hspace{-1mm},\dots,\sqrt 2\mathrm{erf}^{-1}\left(1-4\sqrt{\eta_m}\mathcal{K}_1\left(2\sqrt{\eta_m}\right)\right);\vartheta_\mathrm{u_n}\right)\right]\\
		&\times \Phi_{\vec{R}_\mathrm{e}}\left(\sqrt 2\mathrm{erf}^{-1}\left(1-4\sqrt{\tau_m}\mathcal{K}_1\left(2\sqrt{\tau_m}\right)\right)\hspace{-1mm},\dots,\sqrt 2\mathrm{erf}^{-1}\left(1-4\sqrt{\tau_m}\mathcal{K}_1\left(2\sqrt{\tau_m}\right)\right);\vartheta_\mathrm{e}\right)\label{eq-asc-un}
	\end{align}
	\hrulefill
		\begin{align}\notag
		\overline{\mathcal{C}}_{\mathrm{s,u_f}}^\mathrm{ext} \approx &\frac{p_\mathrm{u_f}}{p_\mathrm{u_n}2\ln 2 }\sum_{\tilde{m}=0}^{\tilde{M}} \frac{\mathrm{e}^{\psi_m}\tilde{\omega}_{\tilde{m}}}{1+\psi_{\tilde{m}}}\left[1-\Phi_{\vec{R}_\mathrm{u_f}}\left(\sqrt 2\mathrm{erf}^{-1}\left(1-4\sqrt{\xi_{\tilde{m}}}\mathcal{K}_1\left(2\sqrt{\xi_{\tilde{m}}}\right)\right)\hspace{-1mm},\dots,\sqrt 2\mathrm{erf}^{-1}\left(1-4\sqrt{\xi_{\tilde{m}}}\mathcal{K}_1\left(2\sqrt{\xi_{\tilde{m}}}\right)\right);\vartheta_\mathrm{u_f}\right)\right]\\
		&\times \Phi_{\vec{R}_\mathrm{e}}\left(\sqrt 2\mathrm{erf}^{-1}\left(1-4\sqrt{\chi_{\tilde{m}}}\mathcal{K}_1\left(2\sqrt{\chi_{\tilde{m}}}\right)\right)\hspace{-1mm},\dots,\sqrt 2\mathrm{erf}^{-1}\left(1-4\sqrt{\chi_{\tilde{m}}}\mathcal{K}_1\left(2\sqrt{\chi_{\tilde{m}}}\right)\right);\vartheta_\mathrm{e}\right)\label{eq-asc-uf}
	\end{align}
	\hrulefill
	\end{figure*}
	\end{theorem}
	\begin{proof}
	By extending \eqref{eq-asc-def1}, we have \cite{ghadi2024_fas_pls}
	\begin{align}
	\overline{\mathcal{C}}_{\mathrm{s},i}^\mathrm{ext} = \frac{1}{\ln 2}\int_0^\infty \frac{\overline{F}_{\gamma_i}\left(\gamma_i\right)}{1+\gamma_i}F_{\gamma_{\mathrm{e},i}}\left(\gamma_i\right)\mathrm{d}\gamma_i,\label{eq-p4} 
\end{align}
	where $\overline{F}_{\gamma_i}\left(\gamma_i\right)=1-F_{\gamma_i}\left(\gamma_i\right)$ represents the complementary CDF (CCDF) of $\gamma_\mathrm{u_n}$ and $F_{\gamma_{\mathrm{e},i}}\left(\gamma_{i}\right)$ is the CDF of $\gamma_{\mathrm{e},i}$. By applying the RV transformation technique to \eqref{eq-cdf}, the CDF expressions of the node $j$ are given by
	\begin{align}
	F_{\gamma_\mathrm{u_n}}\left(\gamma\right) = F_{g_\mathrm{fas,u_n}}\left(\frac{\gamma}{\overline{\gamma}_\mathrm{u_n}p_\mathrm{u_n}}\right),\label{eq-sinr-cdf-n}
	\end{align}
		\begin{align}
		F_{\gamma_\mathrm{u_f}}\left(\gamma\right) = \begin{cases}
			F_{g_\mathrm{fas,u_f}}\left(\frac{\gamma}{\overline{\gamma}_\mathrm{u_f}\left(p_\mathrm{u_f}-\gamma p_\mathrm{u_n}\right)}\right), & 0 \leq \gamma < \frac{p_\mathrm{u_f}}{p_\mathrm{u_n}},\\
			1, & \gamma\ge \frac{p_\mathrm{u_f}}{p_\mathrm{u_n}},
			\end{cases}\label{eq-sinr-cdf-f}
	\end{align}
	\begin{align}
	F_{\gamma_\mathrm{e,u_n}}\left(\gamma\right)=F_{g_\mathrm{fas,e}}\left(\frac{\gamma}{\overline{\gamma}_\mathrm{e}p_\mathrm{u_n}}\right),
	\end{align}
	and
		\begin{align}
		F_{\gamma_\mathrm{e,u_f}}\left(\gamma\right)=F_{g_\mathrm{fas,e}}\left(\frac{\gamma}{\overline{\gamma}_\mathrm{e}p_\mathrm{u_f}}\right).
	\end{align}
	After inserting \eqref{eq-sinr-cdf-n} and \eqref{eq-sinr-cdf-f} into \eqref{eq-p4}, $\overline{\mathcal{C}}_{\mathrm{s},i}^\mathrm{ext}$ can be rewritten as
	\begin{align}	\overline{\mathcal{C}}_{\mathrm{s,u_n}}^\mathrm{ext} = \frac{1}{\ln 2}\int_0^\infty \frac{1-F_{g_\mathrm{fas,u_n}}\left(\frac{\gamma}{\overline{\gamma}_\mathrm{u_n}p_\mathrm{u_n}}\right)}{1+\gamma}F_{g_\mathrm{fas,e}}\left(\frac{\gamma}{\overline{\gamma}_\mathrm{e}p_\mathrm{u_n}}\right)\mathrm{d}\gamma,\label{eq-p5}
	\end{align}
	and
	\begin{align}\notag
	\overline{\mathcal{C}}_{\mathrm{s,u_f}}^\mathrm{ext} = &\frac{1}{\ln 2}\int_0^{\frac{p_\mathrm{uf}}{p_\mathrm{u_n}}} \frac{1-F_{g_\mathrm{fas,u_f}}\left(\frac{\gamma}{\overline{\gamma}_\mathrm{u_f}\left(p_\mathrm{u_f}-\gamma p_\mathrm{u_n}\right)}\right)}{1+\gamma}\\
	&\hspace{3cm}\times F_{g_\mathrm{fas,e}}\left(\frac{\gamma}{\overline{\gamma}_\mathrm{e}p_\mathrm{u_n}}\right)\mathrm{d}\gamma. \label{eq-p6}
	\end{align}
	It is clear that solving \eqref{eq-p5} and \eqref{eq-p6} directly is mathematically complex. Therefore, by applying Lemma \ref{lemma-glq} to \eqref{eq-p5}, $\overline{\mathcal{C}}_{\mathrm{s,u_n}}^\mathrm{ext}$ is obtained as \eqref{eq-asc-un}, completing the proof. For $\overline{\mathcal{C}}_{\mathrm{s,u_f}}^\mathrm{ext}$, we utilize the Gauss-Legendre quadrature method, which is defined by the following lemma.
	\begin{lemma}\label{lemma2}
		The Gauss-Legendre quadrature is defined as \cite{Abromowitz1972handbook}
		\begin{align}
		\int_a^b \Lambda\left(x\right)\approx \frac{b-a}{2}\sum_{\tilde{m}}^{\tilde{M}}\tilde{\omega}_{\tilde{m}}\Lambda\left(\frac{b-a}{2}\psi_{\tilde{m}}+\frac{b+a}{2}\right),
		\end{align}
		where $\tilde{M}$ is the number of sample points used, $\psi_{\tilde{m}}$  is the $\tilde{m}$-th root of Legendre polynomial $P_{\tilde{m}}\left(\psi_{\tilde{m}}\right)$, and $\tilde{\omega}_{\tilde{m}}=\frac{2}{\left(1-\psi_{\tilde{m}}^2\right)\left[P'_{\tilde{m}}\left(\psi_{\mathrm{m}}\right)\right]^2}$  denotes the corresponding quadrature weight.
	\end{lemma} 
	Now, by applying  Lemma \ref{lemma2}, \eqref{eq-asc-uf} is derived and the proof is accomplished. 
	\end{proof}
	\subsubsection{Internal Eavesdropping} Based on the secrecy capacity definition in \eqref{eq-sc-def2}, the ASC under the internal eavesdropping scenario is given by
	\begin{align}
	\overline{\mathcal{C}}_{\mathrm{s},i}^\mathrm{int}=\int_0^\infty\int_0^{\gamma_{\mathrm{u_n}}}\mathcal{C}_{\mathrm{s},i}^\mathrm{ext}\left(\gamma_i,\gamma_{\mathrm{e},i}\right)f_{\gamma_i}\left(\gamma_i\right)f_{\gamma_{\mathrm{e},i}}\left(\gamma_{\mathrm{e},i}\right)\mathrm{d}\gamma_i\mathrm{d}\gamma_{\mathrm{e},i}.\label{eq-pp1}
	\end{align}
	\begin{theorem}
	The ASC expression of the near user $\mathrm{u_n}$ and the far user $\mathrm{u_n}$ in the proposed FAS-aided secure WPCN  under the internal eavesdropping scenario is given by \eqref{eq-asc-int}, where $\eta_m = \frac{\epsilon_m}{\overline{\gamma}_\mathrm{u_n}p_\mathrm{u_n}}$ and $\upsilon_m=\frac{\epsilon_m}{	\overline{\gamma}_\mathrm{u_f}p_\mathrm{u_n}}$.
	\begin{figure*}
	\begin{align}\nonumber
	\overline{\mathcal{C}}_{\mathrm{s,u_n}}^\mathrm{int} \approx&\frac{1}{\ln 2}\sum_{m=0}^M \frac{\mathrm{e}^{\epsilon_m}\omega_m}{1+\epsilon_m}\left[1-\Phi_{\vec{R}_\mathrm{u_n}}\left(\sqrt 2\mathrm{erf}^{-1}\left(1-4\sqrt{\eta_m}\mathcal{K}_1\left(2\sqrt{\eta_m}\right)\right)\hspace{-1mm},\dots,\sqrt 2\mathrm{erf}^{-1}\left(1-4\sqrt{\eta_m}\mathcal{K}_1\left(2\sqrt{\eta_m}\right)\right);\vartheta_\mathrm{u_n}\right)\right]\\
	&\times \Phi_{\vec{R}_\mathrm{u_f}}\left(\sqrt 2\mathrm{erf}^{-1}\left(1-4\sqrt{\upsilon_m}\mathcal{K}_1\left(2\sqrt{\upsilon_m}\right)\right)\hspace{-1mm},\dots,\sqrt 2\mathrm{erf}^{-1}\left(1-4\sqrt{\upsilon_m}\mathcal{K}_1\left(2\sqrt{\upsilon_m}\right)\right);\vartheta_\mathrm{u_f}\right)\label{eq-asc-int}
	\end{align}
	\hrulefill
	\end{figure*}
	\end{theorem}
	\begin{proof}
		By extending \eqref{eq-pp1}, we have
		\begin{align}
		\overline{\mathcal{C}}_{\mathrm{s,u_n}}^\mathrm{int} = \frac{1}{\ln 2}\int_0^\infty \frac{\overline{F}_{\gamma_\mathrm{u_n }}\left(\gamma_\mathrm{u_n}\right)}{1+\gamma_\mathrm{u_n}}F_{\gamma_{\mathrm{u_f,n}}}\left(\gamma_\mathrm{u_n}\right)\mathrm{d}\gamma_\mathrm{u_n},\label{eq-p7} 
	\end{align}
	where $F_{\gamma_{\mathrm{u_f,n}}}\left(\gamma_\mathrm{u_n}\right)$ denotes the CDF of $\gamma_{\mathrm{u_f,n}}$, which is derived using the RV transformation technique to \eqref{eq-sinr-ufn}, i.e., 
			\begin{align}
		F_{\gamma_\mathrm{u_f,n}}\left(\gamma\right)=F_{g_\mathrm{fas,u_f}}\left(\frac{\gamma}{	\overline{\gamma}_\mathrm{u_f}p_\mathrm{u_n}}\right).\label{eq-cdf-ufn}
	\end{align}
Now, by substituting \eqref{eq-cdf-ufn} and \eqref{eq-sinr-cdf-n} in \eqref{eq-p7}, we have
\begin{align}
\overline{\mathcal{C}}_{\mathrm{s,u_n}}^\mathrm{int}\hspace{-1mm} =\frac{1}{\ln 2} \hspace{-1mm}\int_0^\infty \frac{1-F_{g_\mathrm{fas,u_n}}\left(\frac{\gamma}{\overline{\gamma}_\mathrm{u_n}p_\mathrm{u_n}}\right)}{1+\gamma}F_{g_\mathrm{fas,u_f}}\left(\frac{\gamma}{	\overline{\gamma}_\mathrm{u_f}p_\mathrm{u_n}}\right)\mathrm{d}\gamma.\label{eq-p8}
\end{align}
After substituting the corresponding CDFs into \eqref{eq-p8} and applying Lemma \ref{lemma-glq}, \eqref{eq-asc-int} is obtained, completing the proof.
	\end{proof}

\section{Numerical Results}\label{sec-num}
In this section, we evaluate the analytical derivations of the SOP and ASC, validating them through a thorough comparison with Monte Carlo simulations. For this evaluation, we define the channel model parameters as $p_\mathrm{u_n}=0.4$, $p_\mathrm{u_f}=0.6$, $\alpha=3$, $L_\mathrm{p}=1$,  $d_\mathrm{t}=100$m, $=d_\mathrm{u_n}=20$m, $d_\mathrm{u_f}=60$m, $d_\mathrm{e}=100$m, $\sigma^2_i=-90$dBm, $\sigma^2_{\mathrm{e}}=-80$dBm $P_\mathrm{p}=30$dBm, $R_{\mathrm{s},i}=0.5$bps, $N_{\mathrm{e}}=4$, $W_{\mathrm{e}}=1\lambda^2$, $M=\tilde{M}=2$. Additionally, since the analytical derivations are expressed in terms of the joint multivariate normal distribution, we numerically implement these calculations using MATLAB, specifically the $\textsf{copulacdf}$ and $\textsf{copulapdf}$ functions. To further enhance the simulation, we apply the algorithm introduced in \cite{ghadi2024gau} to generate the Gaussian copula within the context of the system model under consideration.

\begin{figure}[!t]
	\centering
\includegraphics[width=0.9\columnwidth]{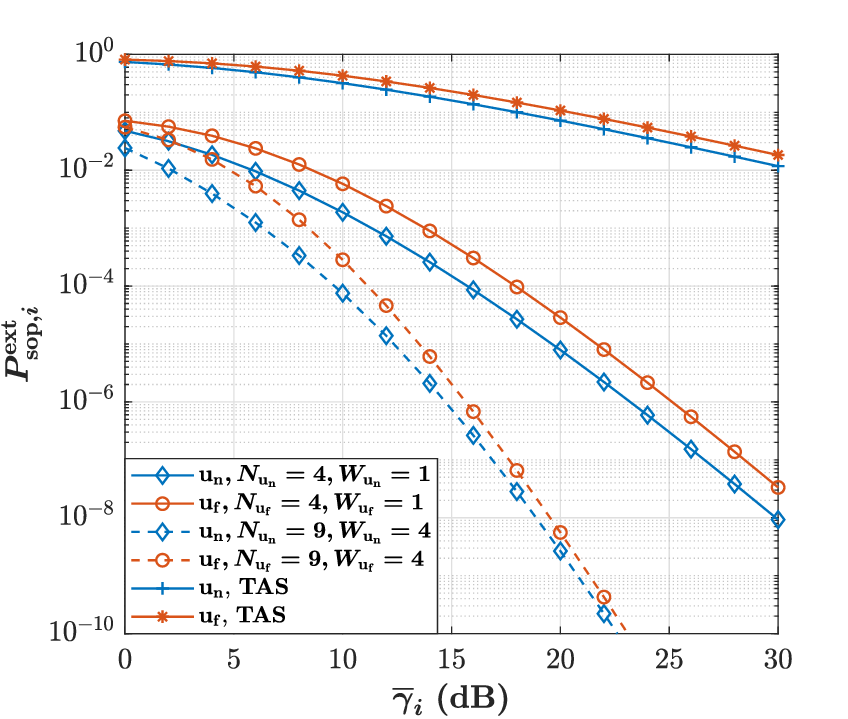}
	\caption{The external SOP of user $i$ versus the average SNR $\overline{\gamma}_i$ when $\overline{\gamma}_\mathrm{e}=0$dB. }\vspace{0cm}
\label{fig-sop_un}
\end{figure}
Fig. \ref{fig-sop_un} indicates the performance of SOP for NOMA user $i$ in terms of the average SNR $\overline{\gamma}_i$ over the external eavesdropping scenario. First of all, it can be observed that the near user $\mathrm{u_n}$, for both FAS and TAS, achieves a lower SOP compared to the far user $\mathrm{u_f}$. The main reason for this is that, despite $\mathrm{u_n}$ having a lower power allocation, it benefits from the SIC process, which allows it to mitigate interference and decode its signal with higher reliability, leading to improved security and, consequently, a lower SOP. Additionally, we observe that when the NOMA user $i$ benefits from the FAS, the SOP is significantly lower compared to when the user is equipped with the TAS. The advantage of FAS lies in its ability to dynamically adjust the antenna configuration, which improves signal quality and enhances system performance, thereby reducing the likelihood of secrecy outage. For instance, we observe that the SOP of the near FAS-equipped NOMA user $\mathrm{u_n}$ with $N_\mathrm{u_n}=9$ and $W_\mathrm{u_n}=4\lambda^2$ is on the order of $10^{-6}$ when $\overline{\gamma}_\mathrm{u_n} = 15$ dB, whereas the SOP for the same user under TAS is approximately $10^{-1}$. It is also observed that the by simultaneously increasing $N_i$ and $W_i$, spatial correlation between fluid antenna
ports becomes balanced, and consequently lower SOP is obtained. 

\begin{figure}[!t]
	\centering
\includegraphics[width=0.9\columnwidth]{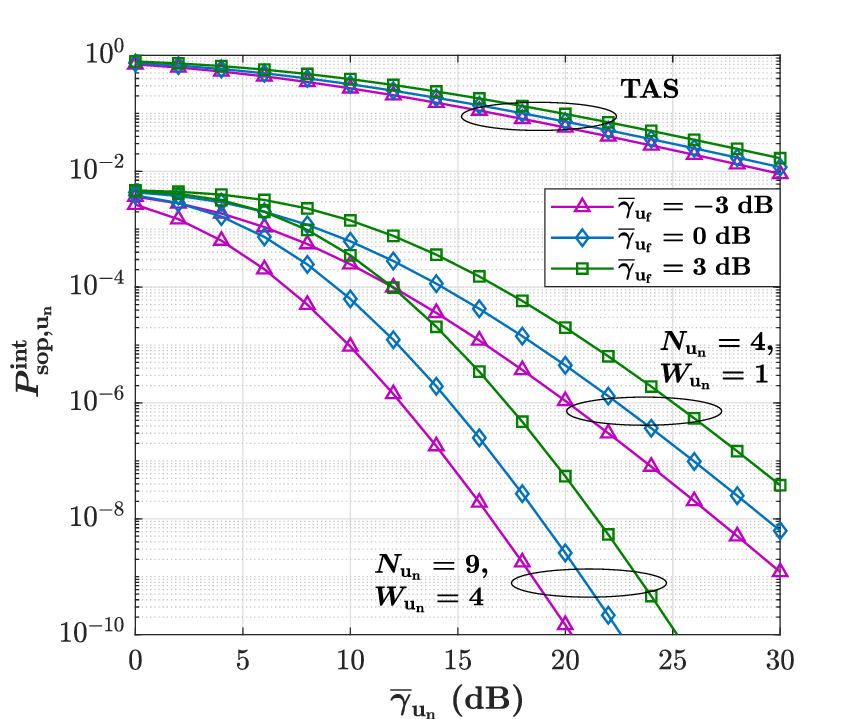}
	\caption{The internal SOP of user $\mathrm{u_n}$ versus the average SNR $\overline{\gamma}_{\mathrm{u_n}}$ when $N_{\mathrm{u_f}}=4$ and $W_{\mathrm{u_f}}=1\lambda^2$.}\vspace{0cm}
\label{fig-sop_int}
\end{figure}
Fig. \ref{fig-sop_int} illustrates the behavior of the SOP in terms of $\overline{\gamma}_\mathrm{u_n}$ when the far user $\mathrm{u_f}$ acts as an internal eavesdropper. We can see for a fixed value of $\overline{\gamma}_{\mathrm{u_f}}$, the SOP of the near user $u_\mathrm{u_n}$ decreases as $\overline{\gamma}_{\mathrm{u_n}}$ increases since the near NOMA user experience better channel condition than the internal eavesdropper (i.e., the far user). Furthermore, it can be observed that increasing $\overline{\gamma}_{\mathrm{u_f}}$ leads to a higher SOP, which is reasonable since the channel conditions of the internal eavesdropper improve.
Additionally, by comparing Fig. \ref{fig-sop_un} and Fig. \ref{fig-sop_int}, we can observe that the SOP of the near user $\mathrm{u_n}$ in the internal eavesdropping scenario performs better than in the external eavesdropping case. The main reason behind this is that in the internal eavesdropping scenario, the far user $\mathrm{u_f}$ (acting as the eavesdropper) typically experiences a weaker signal due to their distance from the transmitter, which makes it more difficult for them to decode the message. In contrast, an external eavesdropper, who may be located in a more favorable position with respect to the signal transmission, could intercept the signal more easily, leading to a higher likelihood of successful eavesdropping. Furthermore, power control strategies and PLS techniques can be more effectively optimized when the eavesdropper is internal, allowing the system to enhance the near user's security by minimizing the interference and signal leakage to the far user. Therefore, although the impact of external and internal eavesdroppers on the near user's SOP can vary significantly depending on specific factors, such as power control strategies and the effectiveness of PLS techniques, in most typical scenarios, the near user's SOP is lower (indicating more secure transmission) under internal eavesdropping compared to external eavesdropping. 

\begin{figure}[!t]
	\centering
\includegraphics[width=0.9\columnwidth]{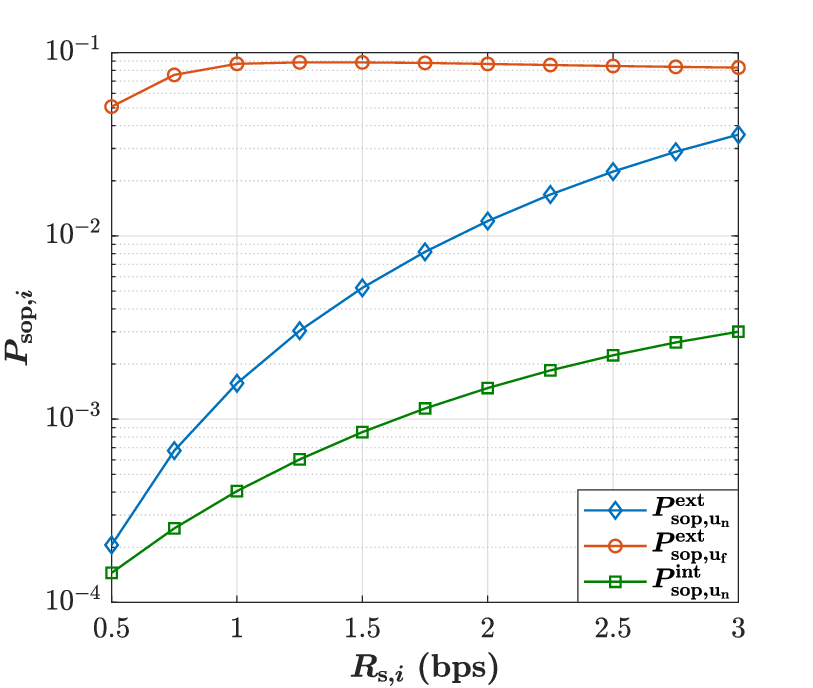}
	\caption{The SOP of user $i$ versus the secrecy rate $R_{\mathrm{s},i}$ when $N_{j}=4$ and $W_{j}=1\lambda^2$.}\vspace{0cm}
\label{fig-sop_rate}
\end{figure}
Fig. \ref{fig-sop_rate} indicates the the relationship between the SOP and secrecy rate $R_{\mathrm{s},i}$ for the FAS-equipped NOMA user $i$. It is observed  that when  $R_{\mathrm{s},i}$ is very low, it becomes easier for the secrecy capacity to exceed the target secrecy rate, resulting  in low SOP for both the near user and the far user. Specifically, since the near user has a strong legitimate channel due to its proximity to the transmitter, ensuring a high achievable secrecy capacity, the SOP of $\mathrm{u_n}$ is very low in both the external and internal eavesdropping scenarios. The far user $\mathrm{u_f}$ has a weaker legitimate channel because of its distance from the transmitter, however, for small $R_{\mathrm{s,u_f}}$ even the far user's capacity can exceed the threshold, leading to a low SOP for $\mathrm{u_f}$. As the target secrecy rate increases, it becomes more difficult for the secrecy capacity to exceed $R_{\mathrm{s},i}$, causing the SOP to rise for both NOMA users. In this case, while $\mathrm{u_n}$ can still maintain a relatively low SOP due to its strong legitimate channel in both the external and internal eavesdropping scenarios, the SOP will steadily increase as  $R_\mathrm{s,u_n}$ rises  further. For the far user $\mathrm{u_f}$, the SOP rises much faster as $R_\mathrm{s,u_f}$ grows compared to the near user,  because the far user's legitimate channel is weaker. 

\begin{figure}[!t]
	\centering
\includegraphics[width=0.9\columnwidth]{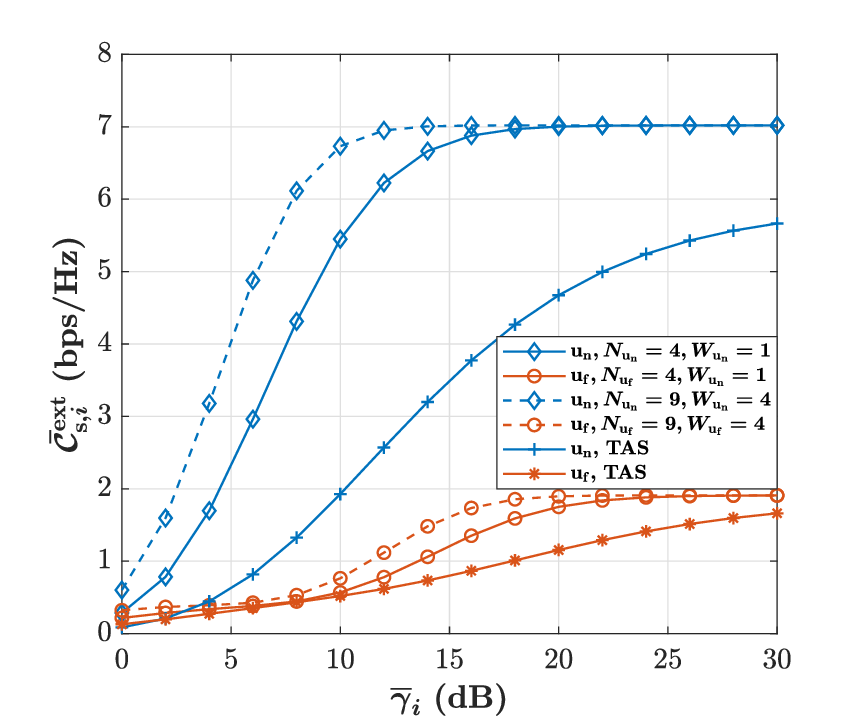}
	\caption{The external ASC of user $i$ versus the average SNR $\overline{\gamma}_i$.}\vspace{0cm}
\label{fig-asc_ext}
\end{figure}
\begin{figure}[!t]
	\centering
\includegraphics[width=0.9\columnwidth]{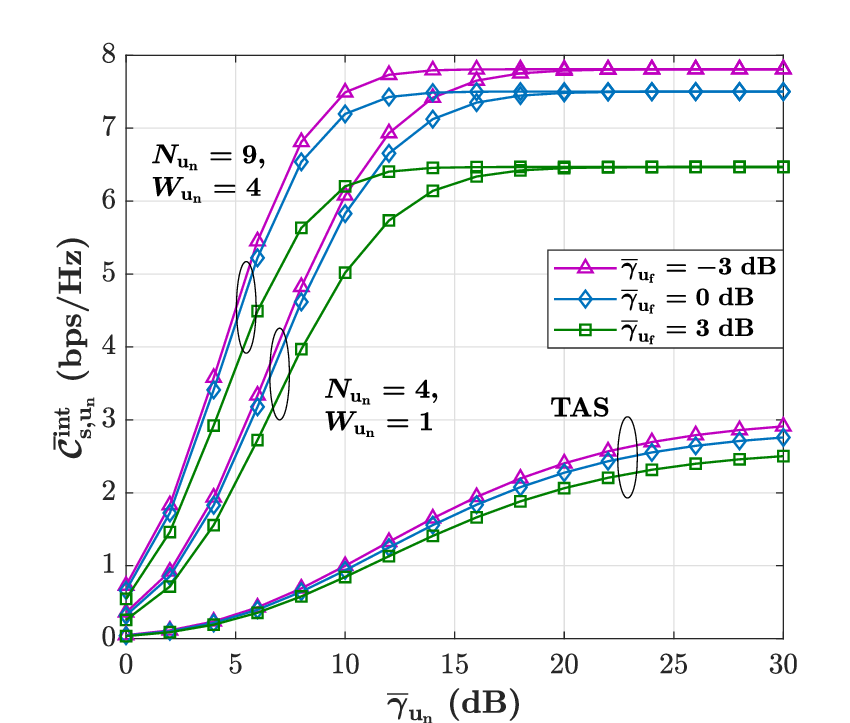}
	\caption{The internal ASC of user $\mathrm{u_n}$ versus the average SNR $\overline{\gamma}_{\mathrm{u_n}}$ when $N_{\mathrm{u_f}}=4$ and $W_{\mathrm{u_f}}=1\lambda^2$.}\vspace{0cm}
\label{fig-asc_int}
\end{figure}
Figs. \ref{fig-asc_ext} and \ref{fig-asc_int} illustrate the performance of the ASC in terms of the average SNR, $\overline{\gamma}_i$, under the external and internal eavesdropping scenarios, respectively. In both scenarios, we observe that as $\overline{\gamma}_i$ increases, the ASC of user $i$ initially improves due to the dominant enhancement in the information link compared to the eavesdropper's link. However, the ASC eventually approaches a saturation point in the high-SNR regime, where further increases in transmit power fail to significantly boost the ASC. This is primarily because the relative difference between the information and eavesdropper channels stabilizes, and additional SNR gains benefit both channels similarly, limiting further improvements in secrecy capacity. Moreover, it is evident that FAS-equipped NOMA users achieve higher ASC values compared to the TAS case in both external and internal eavesdropping scenarios. Additionally, we observe that increasing the number of fluid antenna ports over a larger fluid antenna size results in higher ASC values, as the spatial separation between ports becomes more balanced.


\section{Conclusion}\label{sec-con}
This study focused on evaluating the influence of FAS on the security performance of WPCNs. We proposed a framework where the transmitter relies on energy harvesting from a power beacon to transmit secure information to FAS-enabled users while countering the threats posed by both external and internal eavesdroppers. The use of NOMA added complexity by introducing a dual role for the far user, which could also act as an internal eavesdropper. To address these challenges, we conducted a thorough analysis of the equivalent channel characteristics at each node. Based on these derivations, we developed concise expressions for critical security metrics, such as the SOP and the ASC, using the Gaussian quadrature technique. The results of this research highlight the potential of FAS to enhance security in energy-constrained wireless systems and provide a theoretical foundation for designing robust communication protocols under such scenarios.


\begin{thebibliography}{1}

\bibitem{zhang2022}
H. Zhang, N. Shlezinger, F. Guidi, D. Dardari, M. F. Imani, and Y. C. Eldar, 
``Near-field wireless power transfer for 6G internet of everything mobile networks: Opportunities and challenges,'' 
\textit{IEEE Communications Magazine}, vol. 60, no. 3, pp. 12--18, 2022.

\bibitem{moloudian2024}
G. Moloudian, M. Hosseinifard, S. Kumar, R. B. Simorangkir, J. L. Buckley, C. Song, G. Fantoni, and B. O’Flynn, 
``RF energy harvesting techniques for battery-less wireless sensing, industry 4.0 and internet of things: A review,'' 
\textit{IEEE Sensors Journal}, vol. 24, no. 5, pp. 5732--5745, Mar. 2024.

\bibitem{bi2016}
S. Bi, Y. Zeng, and R. Zhang, 
``Wireless powered communication networks: An overview,'' 
\textit{IEEE Wireless Communications}, vol. 23, no. 2, pp. 10--18, 2016.

\bibitem{wang2023}
X. Wang, J. Li, Z. Ning, Q. Song, L. Guo, S. Guo, and M. S. Obaidat, 
``Wireless powered mobile edge computing networks: A survey,'' 
\textit{ACM Computing Surveys}, vol. 55, no. 13s, pp. 1--37, 2023.

\bibitem{huang2022}
L. Huang, R. Nan, K. Chi, Q. Hua, K. Yu, N. Kumar, and M. Guizani, 
``Throughput guarantees for multi-cell wireless powered communication networks with non-orthogonal multiple access,'' 
\textit{IEEE Transactions on Vehicular Technology}, vol. 71, no. 11, pp. 12104--12116, 2022.

\bibitem{mahmood2022}
A. Mahmood, A. Ahmed, M. Naeem, M. R. Amirzada, and A. Al-Dweik, 
``Weighted utility aware computational overhead minimization of wireless power mobile edge cloud,'' 
\textit{Computer Communications}, vol. 190, pp. 178--189, 2022.

\bibitem{chen2019}
J. Chen, L. Zhang, Y.-C. Liang, X. Kang, and R. Zhang, 
``Resource allocation for wireless-powered IoT networks with short packet communication,'' 
\textit{IEEE Transactions on Wireless Communications}, vol. 18, no. 2, pp. 1447--1461, 2019.

\bibitem{vu2021}
T.-H. Vu and S. Kim, 
``Performance evaluation of power-beacon-assisted wireless-powered NOMA IoT-based systems,'' 
\textit{IEEE Internet of Things Journal}, vol. 8, no. 14, pp. 11655--11665, 2021.

\bibitem{do2020}
D.-T. Do, M.-S. Van Nguyen, T. N. Nguyen, X. Li, and K. Choi, 
``Enabling multiple power beacons for uplink of NOMA-enabled mobile edge computing in wirelessly powered IoT,'' 
\textit{IEEE Access}, vol. 8, pp. 148892--148905, 2020.

\bibitem{han2017}
K. Han and K. Huang, 
``Wirelessly powered backscatter communication networks: Modeling, coverage, and capacity,'' 
\textit{IEEE Transactions on Wireless Communications}, vol. 16, no. 4, pp. 2548--2561, 2017.

\bibitem{ghadi2022}
F. R. Ghadi, F. J. Martin-Vega, and F. J. López-Martínez, 
``Capacity of backscatter communication under arbitrary fading dependence,'' 
\textit{IEEE Transactions on Vehicular Technology}, vol. 71, no. 5, pp. 5593--5598, 2022.

\bibitem{lyu2017}
B. Lyu, Z. Yang, G. Gui, and Y. Feng, 
``Wireless powered communication networks assisted by backscatter communication,'' 
\textit{IEEE Access}, vol. 5, pp. 7254--7262, 2017.

\bibitem{kaveh2023}
M. Kaveh, F. Rostami Ghadi, R. Jäntti, and Z. Yan, 
``Secrecy performance analysis of backscatter communications with side information,'' 
\textit{Sensors}, vol. 23, no. 20, p. 8358, 2023.


\bibitem{ijemaru2022}
G. K. Ijemaru, K. L. M. Ang, and J. K. Seng, 
``Wireless power transfer and energy harvesting in distributed sensor networks: Survey, opportunities, and challenges,'' 
\textit{International Journal of Distributed Sensor Networks}, vol. 18, no. 3, p. 15501477211067740, 2022.

\bibitem{liu2023}
W. Liu, K. T. Chau, X. Tian, H. Wang, and Z. Hua, 
``Smart wireless power transfer—opportunities and challenges,'' 
\textit{Renewable and Sustainable Energy Reviews}, vol. 180, p. 113298, 2023.

\bibitem{kaveh_ris2023}
M. Kaveh, Z. Yan, and R. Jantti, “Secrecy Performance Analysis of RIS-Aided Smart Grid Communications,” \textit{IEEE Transactions on Industrial Informatics}, vol. 20, no. 4, pp. 5415--5427, 2024.

\bibitem{shafiqurrahman2023}
A. Shafiqurrahman, V. Khadkikar, and A. K. Rathore, 
``Vehicle-to-vehicle (V2V) power transfer: Electrical and communication developments,'' 
\textit{IEEE Transactions on Transportation Electrification}, vol. 10, no. 3, pp. 6258--6284, Sep. 2024.


\bibitem{li2022}
M. Li, Y. Zhang, K. Li, Y. Zhang, K. Xu, X. Liu, S. Zhong, and J. Cao, 
``Self-powered wireless sensor system for water monitoring based on low-frequency electromagnetic-pendulum energy harvester,'' 
\textit{Energy}, vol. 251, p. 123883, 2022.


\bibitem{rajak2022}
S. Wang, X. Fu, R. Ruby, and Z. Li, 
``Pilot spoofing detection for massive MIMO mmWave communication systems with a cooperative relay,'' 
\textit{Computer Communications}, vol. 202, pp. 33--41, 2023.

\bibitem{wongfluid2021} K. K. Wong, A. Shojaeifard, K. F. Tong, and Y. Zhang, “Fluid antenna systems,” \textit{IEEE Transactions on Wireless Communications}, Vol. 20, No. 3, pp. 1950-1962, March 2021.

\bibitem{wong2perf2020}K. K. Wong, A. Shojaeifard, K. F. Tong, and Y. Zhang, “Performance limits of fluid antenna systems,” \textit{IEEE Communications Letters}, Vol. 24, No. 11, pp. 2469-2472, November 2020.

\bibitem{wang2024}
C. Wang, Z. Li, K. K. Wong, R. Murch, C. B. Chae, and S. Jin, 
``AI-empowered fluid antenna systems: Opportunities, challenges, and future directions,'' 
\textit{IEEE Wireless Communications}, vol. 31, no. 5, pp. 34--41, Oct. 2024.

\bibitem{new2024_tutorial}
W. K. New, K. K. Wong, H. Xu, C. Wang, F. R. Ghadi, J. Zhang, J. Rao, R. Murch, P. Ramírez-Espinosa, D. Morales-Jimenez, and C. B. Chae, 
``A tutorial on fluid antenna system for 6G networks: Encompassing communication theory, optimization methods and hardware designs,'' 
\textit{IEEE Communications Surveys \& Tutorials}, 2024.

\bibitem{wong2023_part1}
K. K. Wong, W. K. New, X. Hao, K. F. Tong, and C. B. Chae, 
``Fluid antenna system—part I: Preliminaries,'' 
\textit{IEEE Communications Letters}, vol. 27, no. 8, pp. 1919--1923, 2023.

\bibitem{new2023}
W. K. New, K. K. Wong, H. Xu, K. F. Tong, and C. B. Chae, 
``Fluid antenna system: New insights on outage probability and diversity gain,'' 
\textit{IEEE Transactions on Wireless Communications}, vol. 23, no. 1, pp. 128--140, 2023.



\bibitem{ghadi2024}
F. R. Ghadi, M. Kaveh, K. K. Wong, R. Jäntti, and Z. Yan, 
``On performance of FAS-aided wireless powered NOMA communication systems,'' 
in \textit{Proc. 2024 20th Int. Conf. Wireless and Mobile Computing, Networking and Communications (WiMob)}, Oct. 2024, pp. 496--501.



\bibitem{lin2024}
X. Lin, H. Yang, Y. Zhao, J. Hu, and K. K. Wong, 
``Performance analysis of integrated data and energy transfer assisted by fluid antenna systems,'' 
in \textit{Proc. ICC 2024 - IEEE Int. Conf. Commun.}, Jun. 2024, pp. 2761--2766.

\bibitem{ghadi2024_backscatter}
F. R. Ghadi, M. Kaveh, and K. K. Wong, 
``Performance analysis of FAS-aided backscatter communications,'' 
\textit{IEEE Wireless Communications Letters}, vol. 13, no. 9, pp. 2412--2416, 2024.

\bibitem{ye2023}
Y. Ye, L. You, J. Wang, H. Xu, K. K. Wong, and X. Gao, 
``Fluid antenna-assisted MIMO transmission exploiting statistical CSI,'' 
\textit{IEEE Communications Letters}, vol. 28, no. 1, pp. 223--227, Jan. 2024.


\bibitem{jiang2016}
X. Jiang, C. Zhong, X. Chen, T. Q. Duong, T. A. Tsiftsis, and Z. Zhang, 
``Secrecy performance of wirelessly powered wiretap channels,'' 
\textit{IEEE Transactions on Communications}, vol. 64, no. 9, pp. 3858--3871, 2016.

\bibitem{huang2018}
Y. Huang, P. Zhang, Q. Wu, and J. Wang, 
``Secrecy performance of wireless powered communication networks with multiple eavesdroppers and outdated CSI,'' 
\textit{IEEE Access}, vol. 6, pp. 33774--33788, 2018.

\bibitem{cao2023}
K. Cao, H. Ding, L. Lv, Z. Su, J. Tao, F. Gong, and B. Wang, 
``Physical-layer security for intelligent-reflecting-surface-aided wireless-powered communication systems,'' 
\textit{IEEE Internet of Things Journal}, vol. 10, no. 20, pp. 18097--18110, 2023.



\bibitem{ghadi2024_fas_pls}
F. R. Ghadi, K. K. Wong, F. J. López-Martínez, W. K. New, H. Xu, and C. B. Chae, 
``Physical layer security over fluid antenna systems: Secrecy performance analysis,'' 
\textit{IEEE Transactions on Wireless Communications}, vol. 23, no. 12, pp. 18201--18213, Dec. 2024.

\bibitem{vega2024}
J. D. Vega-Sánchez, L. Urquiza-Aguiar, H. R. C. Mora, N. V. O. Garzón, and D. P. M. Osorio, 
``Fluid antenna system: Secrecy outage probability analysis,'' 
\textit{IEEE Transactions on Vehicular Technology}, vol. 73, no. 8, pp. 11458--11469, Aug. 2024.

\bibitem{ghadi2024_ris}
F. R. Ghadi, K. K. Wong, M. Kaveh, F. J. Lopez-Martinez, W. K. New, and H. Xu, 
``Secrecy performance analysis of RIS-Aided fluid antenna systems,'' 
\textit{arXiv preprint}, arXiv:2408.14969, 2024.



\bibitem{new2023_noma}
W. K. New, K. K. Wong, H. Xu, K. F. Tong, C. B. Chae, and Y. Zhang, 
``Fluid antenna system enhancing orthogonal and non-orthogonal multiple access,'' 
\textit{IEEE Communications Letters}, 2023.

\bibitem{tlebaldiyeva2023}
L. Tlebaldiyeva, S. Arzykulov, T. A. Tsiftsis, and G. Nauryzbayev, 
``Full-duplex cooperative NOMA-based mmWave networks with fluid antenna system (FAS) receivers,'' 
in \textit{Proc. 2023 Int. Balkan Conf. Commun. Networking (BalkanCom)}, 2023, pp. 1--6.

\bibitem{zheng2024}
J. Zheng, T. Wu, X. Lai, C. Pan, M. Elkashlan, and K. K. Wong, 
``FAS-assisted NOMA short-packet communication systems,'' 
\textit{IEEE Transactions on Vehicular Technology}, 2024.
\bibitem{noor2023deep}
N. Waqar, K. K. Wong, K. F. Tong, A. Sharples and Y. Zhang, 
``Deep Learning Enabled Slow Fluid Antenna Multiple Access,'' 
\textit{IEEE Communications Letters},  vol. 27, no. 3, pp. 861-865, March 2023.
\bibitem{gradshteyn2007table} I. S. Gradshteyn and I. M. Ryzhik, \textit{Table of integrals, series, and
products.} Academic, 7th ed., 2007.
\bibitem{Nelsen2006intro}
R. B. Nelsen, ``\textit{An introduction to copulas,}" Springer, 2006.
\bibitem{ghadi2024gau}
F. Rostami Ghadi, K. -K. Wong, F. Javier López-Martínez, C. -B. Chae, K. -F. Tong and Y. Zhang, 
``Gaussian Copula Approach to the Performance Analysis of Fluid Antenna Systems,'' 
\textit{IEEE Transactions on Wireless Communications},  vol. 23, no. 11, pp. 17573-17585, Nov. 2024.
\bibitem{Abromowitz1972handbook}
M. Abromowitz and I. A. Stegun, “\textit{Handbook of mathematical functions,}” 1972.








\end{thebibliography}

\end{document}